\documentclass[11pt, letter]{article}
\usepackage{graphicx,fullpage,amsthm,amsmath}
\usepackage{enumerate}
\newtheoremstyle{mythmstyle}
  {}
  {}
  {\slshape}
  {}
  {\bfseries}
  {.}
  {.5em}
  {}

\theoremstyle{mythmstyle}
\newtheorem{theorem}{Theorem}
\newtheorem{definition}[theorem]{Definition}
\newtheorem{lemma}[theorem]{Lemma}

\newcommand\dist[1]{\langle#1\rangle}

\newcommand\bfR{\mathbf{R}}

\newcommand\ap{\overrightarrow{p}}

\title{Peres-Style Recursive Algorithms}
\author{Sung-il Pae\\Hongik University\\Seoul, Korea\\email: pae@hongik.ac.kr}

\begin{document}
\maketitle
  
\begin{abstract}
\noindent
Peres algorithm applies the famous von Neumann's trick recursively to produce
unbiased random bits from biased coin tosses.  Its recursive nature makes the
algorithm simple and elegant, and yet its output rate approaches the
information-theoretic upper bound.  However, it is relatively hard to explain
why it works, and it appears partly due to this difficulty that its
generalization to many-valued source was discovered only recently.
Binarization tree provides a new conceptual tool to understand the
innerworkings of the original Peres algorithm and the recently-found
generalizations in both aspects of the uniform random number generation and
asymptotic optimality.  Furthermore, it facilitates finding many new
Peres-style recursive algorithms that have been arguably very hard to come by
without this new tool.
\end{abstract}
\section{Introduction}
Given a coin that turns heads (denoted by 0) with probability $p$, thus the
probability of turning tails (1) being $q=1-p$, the von Neumann's trick takes
two coin flips and returns output by the following rule \cite{VonNeumann51}:
\begin{equation}\label{eq:vNrule}
00\mapsto \lambda,\; 01\mapsto 0,\; 10\mapsto 1,\; 11\mapsto \lambda,
\end{equation}
where $\lambda$ indicates ``no output.''  Because $\Pr(01)=\Pr(10)=pq$, the
resulting bit is unbiased.  By repeating this process, we obtain a sequence
of random bits, and the {\em output rate\/}, the average number of output per
input, is $pq\le 1/4$. (See, for example, exercise 5.1-3 in
\cite{stein2001introduction})

Formalizing this idea~\cite{elias72,peres92,pae06-randomizing,pae-loui05}, an
{\em extracting function} $f\colon\{0,1\}^n \rightarrow \{0,1\}^*$ takes $n$
independent bits of bias $p$, called Bernoulli source of bias $p$, and
returns independent and unbiased random bits, and its output rate is bounded
by the Shannon entropy $H(p)=-(p\lg p+q\lg q)$.  When $p=1/3$, the output
rate of von Neumann's procedure is $pq=2/9\approx 0.22$ while the entropy
bound $H(1/3)\approx0.92$; the discrepancy is quite large.  But there are
{\em asymptotically optimal\/} extracting functions that achieve rates
arbitrarily close to the entropy bound.

Consider the functions defined on $\{0,1\}^2$ as follows, where $\Psi_1$ is
the {\em von Neumann} function defined by the rule (\ref{eq:vNrule}):
\begin{equation}\label{eq:peres-table}
\begin{tabular}{c|c|c|c|c}
\hline
$x$& $\Pr(x)$&$\Psi_1(x)$&$u(x)$ & $v(x)$\\
\hline
00&$p^2$&$\lambda$&0&0\\
01&$pq$&0&1&$\lambda$\\
10&$pq$&1&1&$\lambda$\\
11&$q^2$&$\lambda$&0&1\\
\hline
\end{tabular}\bigskip
\end{equation}
Extend the three functions $\Psi_1$, $u$, and $v$ to $\{0,1\}^*$: for an
empty string,
\[
\Psi_1(\lambda)=u(\lambda)=v(\lambda)=\lambda,
\]
for a nonempty even-length input, define (and the same for $u$ and $v$)
\begin{equation}\label{eq:vN-ext}
\Psi_1(x_1x_2\dots x_{2n})=\Psi_1(x_1x_2)*\cdots*\Psi_1(x_{2n-1}x_{2n}),
\end{equation}
where $*$ is concatenation, and for an odd-length input, drop the last
bit and take the remaining even-length bits.  

Now, define {\em Peres function\/} $\Psi:\{0,1\}^*\to\{0,1\}^*$ by a
recursion
\begin{equation}\label{eq:peres-def}
\left\{
\begin{split}
\Psi(x) &= \Psi_1(x) * \Psi(u(x)) * \Psi(v(x)),\\
\Psi(\lambda) &=\lambda.
\end{split}
\right.
\end{equation}
This simple recursive function is extracting for each input length, and,
rather surprisingly, asymptotically optimal~\cite{peres92}.  Its
implementation is straightforward and runs very fast, in $O(n\log n)$ time
for input length $n$, with a small footprint.  Another superiority over other
asymptotically optimal extracting algorithms, for example, Elias
algorithm~\cite{elias72}, is its uniformity.  To achieve the asymptotic
optimality, these algorithms need to take increasingly longer inputs.  While
Peres algorithm does this with the same simple fixed function $\Psi$, Elias
algorithm need to compute separately for each input length with increasing
complexity.

However, it appears harder to explain why Peres algorithm works than Elias
algorithm, and it is quite tempting to say that the algorithm works almost
like a magic because of its simplicity of definition and complexity of
justification.  So a natural question is whether we can find similar
recursively defined extracting functions.
But the question had remained elusive for a while and it was only recent
that its generalizations to many-valued source were discovered~\cite{pae15}.

By {\em Peres-style\/} recursion we mean a recursion of the style
\[
\Psi(x)=\Psi_1(x)*\Psi(u_1(x))*\dots*\Psi(u_l(x)),
\]
which defines an asymptotically optimal extracting function, where $\Psi_1$
is extracting and $u_1,\dots,u_l$ are auxiliary functions defined on a fixed
finite-length inputs and extended to arbitrary-length inputs as
in~\eqref{eq:vN-ext}.  If one or more auxiliary functions are omitted, then
the resulting recursive function is still extracting but not asymptotically
optimal anymore.  We always call the base function of the recursion von
Neumann function and write as $\Psi_1$.

We report a new way to understand and justify Peres-style recursive
algorithms using a recent development of {\em binarization
  tree\/}~\cite{DBLP:conf/isit/Pae16}.  It provides a simple and unified
viewpoint that explains innerworkings of the Peres algorithm and its
recently-found generalizations.  Furthermore, we report new Peres-style
recursive algorithms that have been arguably hard to come by without this new
tool.  

\section{Binarization Tree and Peres Algorithm}
Summarized below are necessary backgrounds on extracting functions and
binarization trees.  In particular, for a binarization tree, we give
``structure lemma'' and ``entropy lemma.''  The entropy lemma is also known
as the ``leaf entropy theorem,'' (see, for example, Section 6.2.2. of
\cite{Knuth:art3}) and it is mainly related to the asymptotic optimality of
Peres-style recursive algorithms defined by a binarization tree.  The
structure lemma was first discussed in~\cite{DBLP:conf/isit/Pae16}, and, in
our context, it is used to show our algorithms are extracting.  For more
rigorous treatments on the subjects, see
\cite{pae06-randomizing,pae15,DBLP:conf/isit/Pae16}.  Then, using these new
tools, we give a proof that the original Peres algorithm is extracting and
asymptotically optimal.

\subsection{Extracting Functions}
\begin{definition}[\cite{peres92,pae06-randomizing}]\label{def:extracting}
A function $f\colon\{0,1,\dots,m-1\}^{n}\to \{0,1\}^\ast$ is {\em
  $m$-extracting} if for each pair $z_1,z_2$ in $\{0,1\}^\ast$ such that
$|z_1|=|z_2|$, we have $\Pr(f(x)=z_1)=\Pr(f(x)=z_2)$, regardless of
the distribution $\dist{p_0,\dots,p_{m-1}}$.
\end{definition}
\noindent
Denote by $S_{(n_0,n_1,\dots,n_{m-1})}$ the subset of
$\{0,1,\dots,m-1\}^{n}$ that consists of strings with $n_i$ $i$'s.
Then
\[
\{0,1,\dots,m-1\}^{n}=\bigcup_{n_0+n_1+\dots+n_{m-1}=n} S_{(n_0,n_1,\dots,n_{m-1})},
\]
and each $S_{(n_0,n_1,\dots,n_{m-1})}$ is an {\em equiprobable\/} subset of
elements whose probability of occurrence is $p_0^{n_0}p_1^{n_1}\cdots
p_{m-1}^{n_{m-1}}$.
When $m=2$, an equiprobable set $S_{(l,k)}$ is also written as $S_{n,k}$,
where $n=l+k$, and its size can also be written as an equivalent binomial
coefficient as well as the multinomial one:
\[
{n\choose k}={n\choose{l,k}}.
\]

An equivalent condition for a function to be extracting is that it sends
equiprobable sets to multiple copies of $\{0,1\}^N$, the
exact full set of binary strings of various lengths $N$'s.  
 For example, Table~\ref{tbl:example} shows how von Neumann function and
 Peres function sends equiprobable sets to such sets.
\begin{table}[h]\label{tbl:example}
\begin{center}
\begin{tabular}{c|c|c|c|c||c|c}
\hline
$k$ ($n=6$)& $\Pr(x)$ &$|S_{n,k}|$ & von Neumann ($\Psi_1$)& bits& Peres
($\Psi$) & bits\\
\hline
$k=0$ & $p^6$ & 1 & $\{\lambda\}$ &0& $\{\lambda\}$ &0\\
$k=1$ & $p^5q$ & 6& $3\cdot\{0,1\}$ &6& $\{0,1\},\{0,1\}^2$&10 \\
$k=2$ & $p^4q^2$&15 &$3\cdot\{\lambda\},\, 3\cdot\{0,1\}^2$ & 24&
   $\{\lambda\},\, \{0,1\},\, \{0,1\}^2,\, \{0,1\}^2$&34\\
$k=3$&$p^3q^3$&20&$6\cdot\{0,1\},\,\{0,1\}^3$&28&$\{0,1\}^2,\,3\cdot\{0,1\}^3$&56\\
$k=4$ & $p^2q^4$&15 &$3\cdot\{\lambda\},\, 3\cdot\{0,1\}^2$ &24&
   $\{\lambda\},\, \{0,1\},\, \{0,1\}^2,\, \{0,1\}^2$ &34 \\
$k=5$ & $pq^5$ & 6& $3\cdot\{0,1\}$ &6& $\{0,1\},\{0,1\}^2$ &10\\
$k=6$ & $q^6$ & 1 & $\{\lambda\}$ &0& $\{\lambda\}$&0 \\
\hline
\end{tabular}
\caption{Multiset images of equiprobable sets under extracting functions}
\end{center}
\end{table}

\begin{definition}[\cite{pae15}]
A multiset $A$ of bit strings is {\em extracting} if, for each $z$ that
occurs in $A$, all the bit strings of length $|z|$ occur in $A$ the same time
as $z$ occurs in $A$.
\end{definition}
\begin{lemma}[\cite{pae15}]\label{lemma:extracting-multiset}
A function $f\colon\{0,1,\dots,m-1\}^{n}\to \{0,1\}^\ast$ is extracting if
and only if its multiset image of each equiprobable set
$S_{(n_0,n_1,\dots,n_{m-1})}$ is extracting.
\end{lemma}

\subsection{Binarization Tree}
Let $X$ be a random variable over $\{0, 1,\dots, m-1\}$ (or, rather, a dice
with $m$ faces) with probability distribution $\dist{p_0,\dots,p_{m-1}}$.  A
sequence $x=x_1\dots x_n\in\{0,1,\dots,m-1\}^{n}$ is considered to be taken
$n$ times from $X$.

Given a function $\phi\colon\{0,1,\dots,m-1\}\to\{\lambda,0,1,\dots,k-1\}$,
the random variable $\phi(X)$ has distribution $\dist{\pi_0,\dots,\pi_{k-1}}$,
where
\[
\pi_0=\sum_{\phi(i)=0}p_i/s,\,\dots,\,\pi_{k-1}=\sum_{\Phi(i)=k-1}p_i/s,\,\text{and}\;
 s=\sum_{\phi(i)\not=\lambda}p_i.
\]
Extend $\phi$ to $\{0,1,\dots,m-1\}^n$, by letting, for $x=x_1\dots x_n$,
$\phi(x)=\phi(x_1)*\dots*\phi(x_n)$.  Then, for an equiprobable set
$S=S_{(n_0,\dots,n_{m-1})}$, its image under $\phi$ is equiprobable, that is,
\[
\phi(S)=S_{(l_0,\dots,l_{k-1})},
\]
where
\[
l_0=\sum_{\phi(i)=0}n_i,\dots, l_{k-1}=\sum_{\phi(i)=k-1}n_i.
\]

\newcommand{\leaf}{\mathrm{leaf}}
Consider a tree with $m$ external nodes labeled uniquely with
$0,1,\dots,m-1$.  For an internal node $v$ of degree $k$, define a function
$\phi_v\colon \{0,1,\dots,m-1\}\to\{\lambda,0,1,\dots,k-1\}$ as follows:
\[
\phi_v(x)=\left\{\begin{array}{ll}
  i,&\text{if $x\in\leaf_i(v)$, for $i=0,\dots,k-1$,}\\
  \lambda,&\text{otherwise.}
 \end{array}\right.
\]
where $\leaf_i(v)$ is the set of external nodes on the $i$th subtree of $v$.
When $X$ is also an $m$-valued source, call such a tree an {\em
  $m$-binarization tree\/} over $X$ and $\phi_v$ its {\em component
  function.}

For example, the following tree with 10 external nodes
\begin{equation}\label{eq:example-tree}
\raisebox{-.5\height}{\includegraphics[scale=1]{figs-31.mps}}
\end{equation}
defines the following component functions:
\begin{equation}\label{eq:yet-another-bin-eg}
\medskip
\begin{tabular}{c|c|c|c|c|c}
\hline
$x$& $\Phi_1(x)$&$\Phi_2(x)$ & $\Phi_3(x)$ & $\Phi_4(x)$ & $\Phi_5(x)$\\
\hline
0&2&$\lambda$&0&1&1\\
1&2&$\lambda$&0&0&$\lambda$\\
2&0&0&$\lambda$&$\lambda$&$\lambda$\\
3&2&$\lambda$&2&$\lambda$&$\lambda$\\
4&2&$\lambda$&0&1&0\\
5&0&1&$\lambda$&$\lambda$&$\lambda$\\
6&1&$\lambda$&$\lambda$&$\lambda$&$\lambda$\\
7&2&$\lambda$&1&$\lambda$&$\lambda$\\
8&2&$\lambda$&0&1&2\\
9&2&$\lambda$&0&1&3\\
\hline
\end{tabular}
\end{equation}  
For $S=S_{(n_0,n_1,\dots,n_{m-1})}$, we have
\begin{align*}
\Phi_1(S)&=S_{(n_2+n_5,n_6,n_0+n_1+n_3+n_4+n_7+n_8+n_9)},\\
\Phi_2(S)&=S_{(n_2,n_5)},\\
\Phi_3(S)&=S_{(n_0+n_1+n_4+n_8+n_9,n_7,n_3)},\\
\Phi_4(S)&=S_{(n_1,n_0,n_4,n_8,n_9)},\\
\Phi_5(S)&=S_{(n_4,n_0,n_8,n_9)},
\end{align*}
and the sizes $|S|$ and $|\Phi_i(S)|$ satisfy
\begin{align*}
|S|&={n\choose{n_0,\dots,n_9}}\\
   &={n\choose{n_2+n_5,n_6,n_0+n_1+n_3+n_4+n_7+n_8+n_9}}
     {n_2+n_5\choose{n_2,n_5}}\dots
     {n_4+n_0+n_8+n_9\choose{n_4,n_0,n_8,n_9}}\\
   &=\prod|\Phi_i(S)|.
\end{align*}
In fact, we have a stronger claim.  A proof is given in Appendix.
\begin{lemma}[Structure Lemma]\label{lemma:structure}
Let $\Phi=\{\Phi_1,\dots,\Phi_{M}\}$ be the set of component functions
defined by an $m$-binarization tree.  Then the mapping
$\Phi\colon x\mapsto \Phi(x)= (\Phi_1(x),\dots,\Phi_{M}(x))$ gives a
one-to-one correspondence between an equiprobable subset
$S=S_{(n_0,n_1,\dots,n_{m-1})}$ and
$\Phi_1(S)\times\cdots\times\Phi_{M}(S)$, the Cartesian product of
equiprobable sets $\Phi_j(S)$'s.
\end{lemma}

For a node $v$ of a binarization tree $T$ and its degree is $k$, let
\[
P(v)=\sum_{i\in\leaf(v)}p_i,
\]
where $\leaf(v)=\bigcup_{i=0,\dots,k-1} \leaf_i(v)$, and let
\begin{align*}
\pi_0(v)&=\sum_{i\in\leaf_0(v)}p_i/P(v),\\
  &\vdots\\
\pi_{k-1}(v)&=\sum_{i\in\leaf_{k-1}(v)}p_i/P(v),
\end{align*}
so that $\phi_v(X)$ has the distribution
$\pi(v)=\dist{\pi_0(v),\dots,\pi_{k-1}(v)}$.  Then we have the following
lemma that is also well-known as the leaf entropy theorem.  See, for example,
Lemma E in Section 6.2.2 of~\cite{Knuth:art3}.
\begin{lemma}[Entropy Lemma]\label{lemma:entropy}
\[
H(X)=\sum_{v\in T}P(v)H(\pi(v)).
\]
\end{lemma}

\subsection{Peres Algorithm Revisited}
Let $Y$ be a Bernoulli random variable with distribution $\dist{p,q}$.
Consider the following binarization tree over $X=Y^2$, the random variable
with values $\{0,1\}^2=\{00,01,10,11\}$ and distribution
$\dist{p^2,pq,pq,q^2}$, and:
\begin{equation}\label{eq:original-peres-tree}
\raisebox{-.5\height}{\includegraphics[scale=1]{figs-11.mps}}
\end{equation}
Then the component functions $\{u,v,\Psi_1\}$ defined by this binarization
tree are exactly the same as those of Peres algorithm given in
\eqref{eq:peres-table}!

\begin{theorem}\label{thm:peres-extracting}
Peres function $\Psi$ is extracting.
\end{theorem}
\begin{proof}
Observe that, for an equiprobable set $S\subset\{0,1\}^{2n}$,
$\Psi(S)=\Psi_1(S)*\Psi(u(S))*\Psi(v(S))$.  This does not hold in general.
But, for equiprobable sets, we have one-to-one correspondence $\Phi$ given by
the binarization tree~\eqref{eq:original-peres-tree} and
$\Phi(S)=\Psi_1(S)\times u(S)\times v(S)$ by the structure lemma.  Consider a
function $\Psi'$ on $\{0,1\}^*\times\{0,1\}^*\times\{0,1\}^*$ defined by
$\Psi'(x,u,v)=x*\Psi(u)*\Psi(v)$.  For sets $A,B,$ and $C$, we have
$\Psi'(A\times B\times C)=A*\Psi(B)*\Psi(C)$.  Since $\Psi=\Psi'\circ\Phi$,
we conclude that $\Psi(S)=\Psi_1(S)*\Psi(u(S))*\Psi(v(S))$.

Note, here, that $u(S)$ and $v(S)$ are equiprobable.  Now, by the induction
on the length of strings, $\Psi(u(S))$ and $\Psi(v(S))$ are extracting.
Since $\Psi_1$ is extracting, so is $\Psi_1(S)$.  So, their concatenation
$\Psi(S)$ is extracting, and thus $\Psi$ is extracting.
\end{proof}

\newcommand\EX{\mathrm{E}} 
\begin{theorem}\label{thm:peres-optimal}
Peres function $\Psi$ is asymptotically optimal.
\end{theorem}
\begin{proof}
By the entropy lemma,
\[
H(Y^2)=2pqH(\Psi_1(Y^2))+H(u(Y^2))+(p^2+q^2)H(v(Y^2)).
\]
The nodes of our binarization tree have distributions
\begin{align*}
u(Y^2)&:\quad \dist{p^2+q^2,2pq},\\
v(Y^2)&:\quad      \dist{p^2/(p^2+q^2), q^2/(p^2+q^2)},\\
\Psi_1(Y^2)&:\quad \dist{\frac12,\frac12}.                              
\end{align*}
Since $H(Y^2)=2H(p)$ and $H(\Psi_1(Y^2))=1$, we have
\begin{equation}\label{eq:original-peres-entropy}
H(p)=pq+\frac12 H(p^2+q^2)+\frac12(p^2+q^2)H(p^2/(p^2+q^2)).
\end{equation}

Consider the truncated versions of Peres function, whose recursion depths are
bounded by $\nu$, defined as follows:
\begin{align*}
\Psi_\nu(x)&=\Psi_1(x)*\Psi_{\nu-1}(u(x))*\Psi_{\nu-1}(v(x)),\\
\Psi_0(x)&=\lambda.
\end{align*}
So the von Neumann function $\Psi_1$ has recursion depth 1, and if
$|x|\le 2^\nu$, then $\Psi(x)=\Psi_\nu(x)$.

The output rate $r_\nu(p)$, for the source distribution $\dist{p,q}$, of
$\Psi_\nu$ satisfies the recursion~(See
\cite{peres92,pae15,DBLP:journals/ipl/Pae13})
\begin{align}\label{eq:original-peres-rate-rec}
 r_\nu(p) = r_1(p) + \frac12 r_{\nu-1}(p^2+q^2) 
  + \frac12 (p^2+q^2) r_{\nu-1}(p^2/(p^2+q^2)),
\end{align}
Note, here, that $u(Y^2)$ and $v(Y^2)$ has distributions $\dist{p^2+q^2,2pq}$
and $\dist{p^2/(p^2+q^2), q^2/(p^2+q^2)}$, respectively, and $r_1(p)=pq$, and
$r_0(p)=0$.  Consider the operator $T$ on $\{f:[0,1]\to\bfR\mid
\lim_{t\to0}f(t)=\lim_{t\to1}f(t)=0\}$ defined by
\begin{equation}\label{eq:original-peres-op}
T(f)(p)=r_1(p)+ \frac12 f(p^2+q^2) 
  + \frac12 (p^2+q^2) f(p^2/(p^2+q^2)).
\end{equation}
Then $r_\nu(p)=T^{(\nu)}(r_0)(p)$ is increasing and bounded by $H(p)$.  By
\eqref{eq:original-peres-entropy}, $H(p)$ is a fixed point of $T$ and thus
$\lim_{\nu\to\infty}r_v(p)=H(p)$.
\end{proof}

In the rest of the paper, for each Peres-style recursion, we give a
binarization tree whose component functions define the von Neumann function
$\Psi_1$ and auxiliary functions $u_1,\dots,u_l$.  The resulting recursive
function $\Psi(x)=\Psi_1(x)*\Psi(u_1(x))*\dots*\Psi(u_l(x))$ is extracting
exactly in the same manner as in Theorem~\ref{thm:peres-extracting}; an
equiprobable set $S$ is sent to an extracting multiset $\Psi_1(S)$ and the
Cartesian product of equiprobable sets $u_1(S)\times\dots\times u_l(S)$ which
in turn becomes extracting multiset $\Psi(u_1(S))*\dots*\Psi(u_l(S))$ so
that $\Psi(S)$ is extracting.

For the asymptotical optimality, the operator $T$ is again defined by the
binarization tree to be
\begin{equation}\label{eq:T-def}
T(f)(\ap)=r_1(\ap)+ P(u_1) f(\pi(u_1))+\dots+ P(u_l) f(\pi(u_l)),
\end{equation}
where $r_1(\ap)$ is the output rate $P(\Psi_1)$ of the von Neumann function and
$P(u_i)$ is the probability for the node corresponding to the component
function $u_i$ and $\pi(u_i)$ is the node's branching probability
distribution.  In the same manner as in Theorem~\ref{thm:peres-optimal}, the
output rates of the truncated recursive functions give rise to a monotone
increasing sequence converging to the Shannon entropy because the entropy is
the fixed point of the operator $T$ by the entropy lemma since $T$ is defined
by the binarization tree.

\section{Peres-Style Recursive Algorithms}
Section~\ref{sec:three-face} and \ref{sec:four-face} present the binarization
trees for the recently-found generalizations of Peres algorithm
in~\cite{pae15}, and the rest of the paper discusses brand-new Peres-style
recursive algorithms.
\subsection{3-Face Peres Algorithm}\label{sec:three-face}
Suppose that we want to find a generalization of Peres algorithm that works
on a 3-faced source $Y$ with a distribution $\dist{p,q,r}$.  As in the original
Peres algorithm, we take two samples and use the obvious generalization of
von Neumann function which we use the same name $\Psi_1$.  We devise a
binarization tree with $3^2=9$ external nodes, whose component function
includes $\Psi_1$:
\[
\includegraphics[scale=1]{figs-121.mps}
\]
Verify that the component functions are the same as the auxiliary functions
of the 3-faced Peres function given in~\cite{pae15}, except that the
functions $\Psi_{11}$, $\Psi_{12}$ and $\Psi_{13}$ with disjoint supports are
combined, which we denote by
$\Psi_1=\Psi_{11}\oplus \Psi_{12}\oplus \Psi_{13}$:
\begin{center}
\begin{tabular}{c|c|c|c|c|c}
\hline
$x$ & $\Pr(x)$ & $\Psi_1(x)$ & $u(x)$ & $v(x)$ & $w(x)$\\
\hline
00& $p^2$ & $\lambda$ & 0 & 0 & $\lambda$\\
01& $pq$ & 0 & 1 & $\lambda$ & 1\\
02& $pr$ & 0 & 1 & $\lambda$ & 2\\
10& $pq$ & 1 & 1 & $\lambda$ & 1\\
11& $q^2$ & $\lambda$ & 0 & 1 & $\lambda$\\
12& $qr$ & 0 & 1 & $\lambda$ & 0\\
20& $pr$ & 1 & 1 & $\lambda$ & 2\\
21& $qr$ & 1 & 1 & $\lambda$ & 0\\
22& $r^2$ & $\lambda$ & 0 & 2 & $\lambda$\\
\hline
\end{tabular}
\end{center}
Since $\Psi_{11}$, $\Psi_{12}$ and $\Psi_{13}$ are extracting and have
disjoint supports, $\Psi_1(x)=\Psi_{11}(x)*\Psi_{12}(x)*\Psi_{13}(x)$, and thus
$\Psi_1$ is extracting.  Then the resulting recursive function
$\Psi(x)=\Psi_1(x)*\Psi(u(x))*\Psi(v(x))*\Psi(w(x))$ is extracting and
asymptotically optimal, where the entropy bound is $H(p,q,r)$.

\subsection{4-Face Peres Algorithm}\label{sec:four-face}
A 4-Face Peres function is given in \cite{pae15}, and it is defined by the
following binarization tree of $4^2=16$ external nodes:
\[
\includegraphics[scale=1]{figs-13.mps}
\]
whose component functions are as follows, where
$\Psi_1=\Psi_{11}\oplus\Psi_{12}\oplus\dots\Psi_{16}$, again, as in
\cite{pae15}:
\begin{center}
\begin{tabular}{c|c|c|c|c|c}
\hline
$x$ & $\Psi_1(x)$ & $u(x)$ & $v(x)$ & 
   $w_1(x)$ & $w_2(x)$ \\
\hline
00& $\lambda$ & 0 & 0 & $\lambda$ & $\lambda$\\
01& 0 & 1 & $\lambda$ & 0 & $\lambda$ \\
02& 0 & 1 & $\lambda$ & 1 & $\lambda$ \\
03& 0 & 1 & $\lambda$ & 2 & $\lambda$ \\
10& 1 & 1 & $\lambda$ & 0 & $\lambda$ \\
11& $\lambda$ & 0 & 1 & $\lambda$ & $\lambda$\\
12& 0 & 1 & $\lambda$ & 3 & $\lambda$ \\
13& 0 & 2 & $\lambda$ & $\lambda$ & 0 \\
20& 1 & 1 & $\lambda$ & 1 & $\lambda$ \\
21& 1 & 1 & $\lambda$ & 3 & $\lambda$ \\
22& $\lambda$ & 0 & 2 & $\lambda$ & $\lambda$\\
23& 0 & 2 & $\lambda$ & $\lambda$ & 1 \\
30& 1 & 1 & $\lambda$ & 2 & $\lambda$ \\
31& 1 & 2 & $\lambda$ & $\lambda$ & 0 \\
32& 1 & 2 & $\lambda$ & $\lambda$ & 1 \\
33& $\lambda$ & 0 & 3 & $\lambda$ & $\lambda$\\
\hline
\end{tabular}
\end{center}

Alternatively, consider, for example,
\[
\includegraphics[scale=1]{figs-131.mps}
\]
and the corresponding component functions are as follows:
\begin{center}
\begin{tabular}{c|c|c|c|c|c}
\hline
$x$ & $\Psi_1(x)$ & $u(x)$ & $v(x)$ & 
   $w_1(x)$ & $w_2(x)$ \\
\hline
00& $\lambda$ & 0 & 0 & $\lambda$ & $\lambda$\\
01& 0 & 1 & $\lambda$ &  0 & $\lambda$ \\
02& 0 & 1 & $\lambda$ &  1 & $\lambda$ \\
03& 0 & 1 & $\lambda$ &  2 & $\lambda$ \\
10& 1 & 1 & $\lambda$ &  0 & $\lambda$ \\
11& $\lambda$ & 0 & 1 & $\lambda$ & $\lambda$\\
12& 0 & 1 & $\lambda$ & 3 & $\lambda$ \\
13& 0 & 1 & $\lambda$ & $\lambda$ & 0 \\
20& 1 & 1 & $\lambda$ & 1 & $\lambda$ \\
21& 1 & 1 & $\lambda$ & 3 & $\lambda$ \\
22& $\lambda$ & 0 & 2 & $\lambda$ & $\lambda$\\
23& 0 & 1 & $\lambda$ & $\lambda$ & 1 \\
30& 1 & 1 & $\lambda$ & 2 & $\lambda$ \\
31& 1 & 1 & $\lambda$ & $\lambda$ & 0 \\
32& 1 & 1 & $\lambda$ & $\lambda$ & 1 \\
33& $\lambda$ & 0 & 3 & $\lambda$ & $\lambda$\\
\hline
\end{tabular}
\end{center}

\subsection{3-bit Peres Algorithm}
Now consider a brand-new situation in which $m=2$ but the component functions
are defined on 3 bits instead of 2 bits as with all the examples given above:
\begin{equation}\label{eq:three-bit-peres}
\raisebox{-.5\height}{\includegraphics[scale=1]{figs-14.mps}}
\end{equation}

\[
\begin{tabular}{c|c|c|c|c|c|c|c}
\hline
$x$&$\Pr(x)$&$u(x)$&$v(x)$&$v_1(x)$&$v_2(x)$&$\Psi_1(x)$&$w(x)$\\
\hline
000&$p^3$&0&0&0&$\lambda$&                 $\lambda$&$\lambda$  \\
001&$p^2q$&1&$\lambda$&$\lambda$&$\lambda$& 0&0  \\
010&$p^2q$&1&$\lambda$&$\lambda$&$\lambda$& 1&0  \\
011&$pq^2$&1&$\lambda$&$\lambda$&$\lambda$& 0&1  \\
100&$p^2q$&0&1&$\lambda$&0&                 $\lambda$&$\lambda$   \\
101&$pq^2$&1&$\lambda$&$\lambda$&$\lambda$& 1&1  \\
110&$pq^2$&0&1&$\lambda$&1&                 $\lambda$&$\lambda$  \\
111&$q^3$&0&0&1&$\lambda$&                 $\lambda$&$\lambda$   \\
\hline
\end{tabular}
\]
The three-bit von Neumann function $\Psi_1=\Psi_{11}\oplus\Psi_{12}$ does not
utilize inputs 100 and 110, and the output rate $2(p+q)pq/3=2pq/3$ is strictly
smaller than $pq$ of the two-bit case.  Therefore, even though 3-bit Peres
algorithm is asymptotically optimal, the convergence to the entropy bound
must be slower.

\subsection{4-bit Peres Algorithm}
If we ever wanted to have a 4-bit Peres function in this fashion, then can we
use $E_4$, the Elias function of input size 4 as the base of the recursion?
Note, in the three-bit case, $\Psi_1$ of \eqref{eq:three-bit-peres} is
actually $E_3$.  With $E_4$, among 16 inputs, only 2 inputs, 0000 and
1111, are wasted.  Consider the following binarization tree:
\begin{equation}\label{eq:four-bit-peres}
\raisebox{-.5\height}{\includegraphics[scale=1]{figs-15.mps}}
\vspace{.5cm}
\end{equation}
So we have the following recursion:
\[
\Psi(x)=E_4(x)*\Psi(u(x))*\Psi(v(x))*\Psi(w(x))*\Psi(w_1(x))*\Psi(w_2(x))
\]
The rate of $E_4$ is 
\begin{align*}
\frac14\left(2\cdot4p^3q+(2\cdot4+2)p^2q^2+2\cdot4pq^3\right)&=\frac14pq(8p^2+10pq+8q^2)\\
&=pq(1+p^2+q^2+\frac12pq)\\
&>1.65\cdot pq.
\end{align*}
However, it seems that the convergence is slower than the original Peres
function.  For a fair comparison, we need to see how the original Peres
function on $\{0,1\}^{4n}$.  Consider, for $x\in\{0,1\}^4$,
\[
\Psi^2(x)=\Psi_1(x)*\Psi_1(u(x))*\Psi_1(v(x))*\Psi^2(uu(x))*\Psi^2(vu(x))*\Psi^2(uv(x))*\Psi^2(vv(x)).
\]
The output rate of base part of this recursion is
\begin{align*}
pq+pq(p^2+q^2)+\frac12p^2q^2/(p^2+q^2)&=pq(1+p^2+q^2+\frac12pq/(p^2+q^2))\\
 &>pq(1+p^2+q^2+\frac12pq).
\end{align*}
So the comparison favors the original Peres function.
The following plot compares this rate with that of $E_4$, where the red one
being the rate of $\Psi^2$:
\[
\includegraphics[scale=0.6]{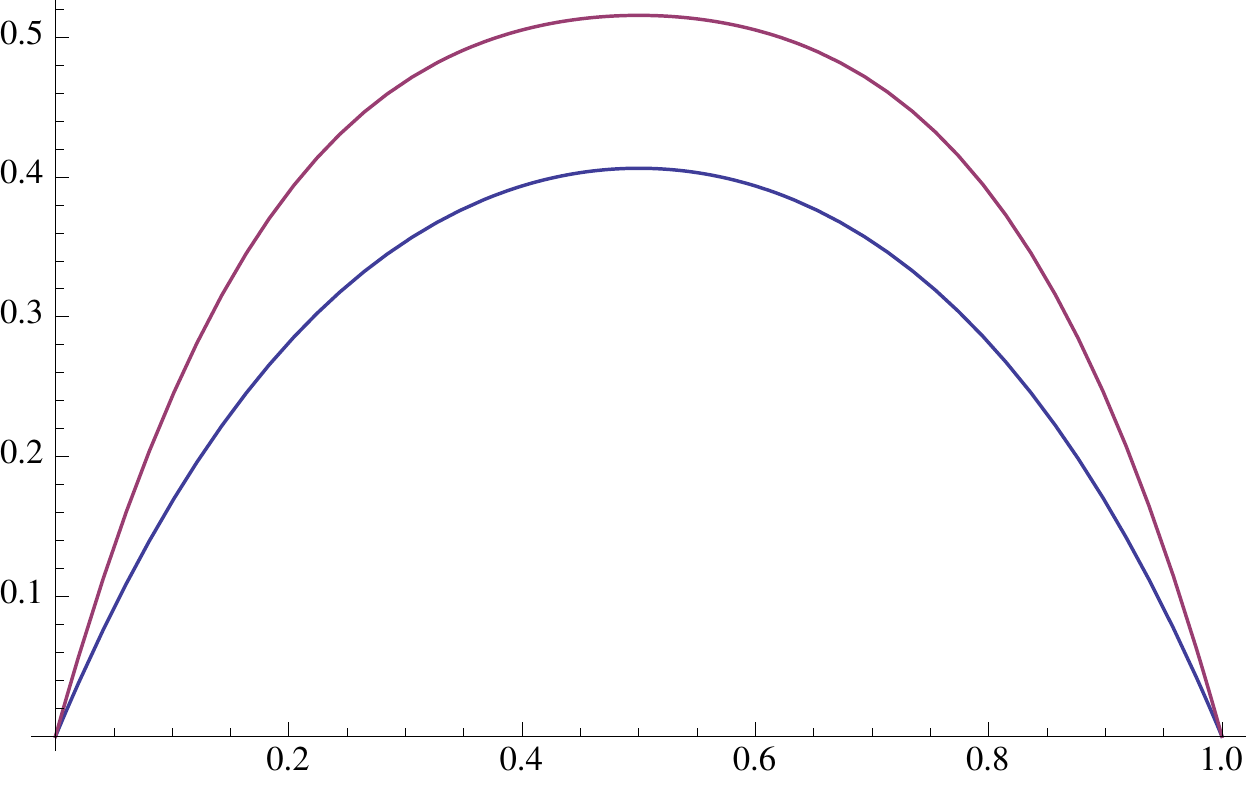}
\]

\subsection{Dijkstra's roulette}
Dijkstra's one-page paper~\cite{dijkstra90} describes an interesting
algorithm that simulates a fair roulette from a biased coin: suppose $m$ is a
prime; take $m$ flips of the coin, encoded as a binary string $x$ in
$\{0,1\}^m$.  If $x=0\dots0$ or $x=1\dots1$, then try again, otherwise,
output $y$ the number of cyclic shifts to obtain the lexicographic minimum.
The virtues of this scheme are, as with Peres algorithm, simplicity and
efficiency, although output rate is much lower than, for example, Elias's
algorithm for $m$-faced dice using a coin, in which case again asymptotically
optimal with output rate approaching $H_m(p)$, the Shannon entropy with base
$m$.  Using Dijkstra's scheme together with Peres-style recursion, we can
reach out the both side of virtues.

Consider a simple case of $m=3$ with a biased coin as a source.  The
Dijkstra's scheme enhanced with Peres-style recursion is described by the
following binarization tree:
\[
\includegraphics[scale=1]{figs-18.mps}
\]

\begin{center}
\begin{tabular}{c|c|c|c|c|c}
\hline
$x$ & $\Pr(x)$ & $\Psi_1(x)$ & $u(x)$ & $v(x)$ & $w(x)$\\
\hline
000& $p^3$ & $\lambda$ & 0 & 0 & $\lambda$\\
001& $p^2q$ & 0 & 1 & $\lambda$ & 0\\
010& $p^2q$ & 1 & 1 & $\lambda$ & 0\\
011& $pq^2$ & 0 & 1 & $\lambda$ & 1\\
100& $p^2q$ & 2 & 1 & $\lambda$ & 0\\
101& $pq^2$ & 2 & 1 & $\lambda$ & 1\\
110& $p^2q$ & 1 & 1 & $\lambda$ & 1\\
111& $q^3$ & $\lambda$ & 0 & 1 & $\lambda$\\
\hline
\end{tabular}
\end{center}
Note, here, that the base function has three branches while auxiliary
functions are binary because we use a coin as the source and outputs
are to be three-valued.  The resulting recursion
$\Psi(x)=\Psi_1(x)*\Psi(u(x))*\Psi(v(x))*\Psi(w(x))$ outputs a uniform
three-valued random numbers with an output rate that approaches $H_3(p)$,
the Shannon entropy with base 3, as the input size tends to infinity.

When $m=5$, consider the following binarization tree:
\begin{equation}\label{eq:tree-m-five}
\raisebox{-.5\height}{\includegraphics[scale=1]{figs-20.mps}}
\end{equation}
where $\Psi_{11}$,\dots,$\Psi_{16}$ are five-valued extracting functions,
which gives the base function $\Psi_1=\Psi_{11}\oplus\dots\oplus\Psi_{16}$.

Now, as $m$ increases, as we can see already in $m=5$ case, the corresponding
binarization tree grows complicated so much that the advantage of the
simplicity disappears quickly.  Note that, in \eqref{eq:tree-m-five}, the
supports of the functions $\Psi_{11},\dots,\Psi_{16}$ are exactly the orbits
with respect to the group action by rotation on
$\{0,1\}^5$~\cite{DBLP:books/daglib/0070572}.  For a prime number $m$, there
are $(2^m-2)/m$ such orbits, and Dijkstra's algorithm is based on this
property.  So the height of the binarization tree grows almost linearly and
the number of nodes exponentially.  However, observe that the subtree rooted
at $w$ in \eqref{eq:tree-m-five} can be regarded as a binary search tree
whose search keys are $\Psi_{11},\dots,\Psi_{16}$.  So we can make a
compromise and keep only the nodes with high probability of output.

For example, for $m=11$, consider the following binarization tree:
\[
\includegraphics[scale=1]{figs-21.mps}
\]
Here, we partition the orbits into four sets $S_1,\dots,S_4$ appropriately,
for example, by the number of 1's.  Then, auxiliary functions $w$, $w_1$ and
$w_2$ are easily computed by counting the number of 1's in the input
$x\in\{0,1\}^{11}$.  The base function $\Psi_1$ is computed using the
original Dijkstra algorithm.  The corresponding Peres-style recursion is
\[
\Psi(x)=\Psi_1(x)*\Psi(u(x))*\Psi(v(x))*\Psi(w_1(x))*\Psi(w_2(x))
\]
is still extracting but not asymptotically optimal.



\section*{Appendix: The Proof of Structure Lemma}
Given a binarization tree, let $T$ be a subtree and $X_T$ be the restriction
of $X$ on the leaf set of $T$.  The leaf entropy theorem is proved by
induction using the following recursion,
\begin{equation}\label{eq:leaf-entropy-rec}
H(X_T)=\begin{cases}
0,& \text{if $T$ is a leaf,}\\
H(\pi)+\pi_0H(X_{T_0})+\dots+\pi_{d-1}H(X_{T_{d-1}}), &\text{otherwise,}
\end{cases}
\end{equation}
where, for nonempty $T$ whose root $v$ has degree $d$, $T_0,\dots,T_{d-1}$ are
the subtrees of $v$ and $\pi=\dist{\pi_0,\dots,\pi_{d-1}}$ is the branching
distribution of $v$.  The structure lemma holds for a similar reason.
\begin{proof}[Proof of Structure Lemma]
For an equiprobable subset $S=S_{(n_0,\dots,n_{m-1})}$ and a subtree $T$ of
the given binarization tree, let $S_T$ be the restriction of $S$ on the leaf
set of $T$.  Then we have a similar recursion
\begin{equation}\label{eq:structure-rec}
S_T\cong\begin{cases}
    \{0\}, & \text{if $T$ is a leaf,}\\
    S_{(l_0,\dots,l_{d-1})}\times S_{T_0}\times\dots\times S_{T_{d-1}},
    &\text{otherwise,}
  \end{cases}
\end{equation}
where, for nonempty $T$ and its root $v$, $T_0,\dots,T_{d-1}$ are the
subtrees of $v$ and
\[
l_0=\sum_{\phi_v(i)=0}n_i,\,\dots,\, l_{d-1}=\sum_{\phi_v(i)=d-1}n_i,
\]
so that $\phi_v(S)=S_{(l_0,\dots,l_{d-1})}$.

First, if $T$ is a leaf with label $i$, then $S_T$ is a singleton set that
consists of a single string of $n_i$ $i$'s, hence the first part of
\eqref{eq:structure-rec}.  When $T$ is nonempty, the correspondence
$S_T\to S_{(l_0,\dots,l_{d-1})}\times S_{T_0}\times\dots\times S_{T_{d-1}}$
is given by $x\mapsto (\phi_v(x),x_{T_0},\dots,x_{T_{d-1}})$, where
$x_{T_0},\dots,x_{T_{d-1}}$ are restrictions of $x$.  This correspondence is
one-to-one because $\phi_v(x)$ encodes the branching with which $x$ is
recovered from $(x_{T_0},\dots,x_{T_{d-1}})$, giving an inverse mapping
$S_{(l_0,\dots,l_{d-1})}\times S_{T_0}\times\dots\times
S_{T_{d-1}}\to S_T$.  For example, consider tree
\eqref{eq:example-tree} and suppose that $T$ is the subtree rooted at the
node {\it 3}.  For $x=207643590289787$, the following shows the restrictions of
$x$ and $\Phi_i(x)$'s.
\[
\medskip
\raisebox{-.5\height}{\includegraphics[scale=1]{figs-310.mps}}
\]
By taking symbols one by one from $x_{T_0}=0490898$, $x_{T_1}=777$, and
$x_{T_2}=3$, according to $\Phi_3(x)=01020000101=(b_i)_{i=1}^{11}$, if
$b_i=j$, from $x_{T_j}$, we recover $x_T=07439890787$.

Induction on subtrees proves the lemma.
\end{proof}


\end{document}